\documentclass[11pt]{article}
\usepackage{authblk}
\usepackage{amsmath,amssymb,amsfonts}
\usepackage{array}
\usepackage{float}
\usepackage{enumerate}
\usepackage{ifthen}
\usepackage{xspace}
\usepackage{graphicx,url}
\usepackage{fullpage}
\usepackage{wasysym}
\usepackage{algorithm}
\usepackage{algpseudocode}

\usepackage[babel=true]{csquotes} 

\renewcommand{\baselinestretch}{1.0}
\newcommand{\qed}{\hspace*{\fill}\rule{6pt}{6pt}\vspace{.5\smallskipamount}}
\newtheorem{theorem}{Theorem}[section]
\newtheorem{definition}{Definition}
\newtheorem{lemma}{Lemma}[section]

\newtheorem{corollary}{Corollary}[section]

\newenvironment{proof} { \noindent \emph{Proof} : } { \qed }

\newenvironment{theorem-repeat}[1]{
\setcounter{theorem}{\ref{#1}}
\addtocounter{theorem}{-1}
\begin{theorem}}
{\end{theorem}}

\newenvironment{lemma-repeat}[1]{
\setcounter{lemma}{\ref{#1}}
\addtocounter{lemma}{-1}
\begin{lemma}}
{\end{lemma}}

\newcounter{linecounter}
\newcommand{\linenumbering}{(\arabic{linecounter})}
\renewcommand{\line}[1]{\refstepcounter{linecounter}
\label{#1}
\linenumbering}
\newcommand{\resetline}{\setcounter{linecounter}{0}}

\newcommand{\hfl}[2]{\ensuremath{HF({#1},{#2})}\xspace}
\newcommand{\confLW}[1]{\ensuremath{\mathcal{L} #1 \mathcal{W}}\xspace}
\newcommand{\secc}{\ensuremath{\textsc{sec}}\xspace}

\newcommand{\remove}[1]{}

\newcommand{\argmax}[1]{\underset{#1}{\operatorname{argmax}}}
\newcommand{\argmin}[1]{\underset{#1}{\operatorname{argmin}}}

\begin{document}


\title{
Wait-Free Gathering of Mobile Robots}

%

\author[1]{Zohir Bouzid}
\author[2]{Shantanu Das}
\author[1]{S\'ebastien Tixeuil}

\affil[1]{University Pierre et Marie Curie - Paris 6, LIP6-CNRS 7606, France.} 
\affil[2]{BGU \& Technion-Israel Institute of Technology, Israel.} 

\date{}
\pagenumbering{arabic}
\maketitle

\begin{abstract}

The problem of gathering multiple mobile robots to a single location, is one of the fundamental problems in distributed coordination between autonomous robots. The problem has been studied and solved even for robots that are anonymous, disoriented, memoryless and operate in the semi-synchronous (ATOM) model. However all known solutions require the robots to be faulty-free except for the results of  Agmon and Peleg \cite{AP06} who solve the gathering problem in presence of one crash fault. This leaves open the question of whether gathering of correct robots can be achieved in the presence of multiple crash failures. We resolve the question in this paper and show how to solve gathering, when any number of robots may crash at any time during the algorithm, assuming \emph{strong} multiplicity detection and chirality. In contrast it is known that for the more stronger byzantine faults, it is impossible to gather even in a 3-robot system if one robot is faulty.
Our algorithm solves the gathering of correct robots in the semi-synchronous model where an adversary may stop any robot before reaching its desired destination. Further the algorithm is self-stabilizing as it achieves gathering starting from any configuration (except the bivalent configuration where deterministic gathering is impossible). 


\end{abstract}

\section{Introduction}

\paragraph*{Robot Networks.}

We consider autonomous robots that are endowed with visibility sensors (but that are otherwise unable to communicate) and motion actuators. Those robots must collaborate to solve a collective task, namely \emph{gathering}, despite being limited with respect to input from the environment, asymmetry, memory, etc. The area where robots have to gather is modeled as a \emph{continuous} two-dimensional Euclidean space, and the gathering task requires every robot to reach a single point that is unknown beforehand, and to remain there hereafter.

Robots operate in \emph{cycles} that comprise \emph{look}, \emph{compute}, and \emph{move} phases. The look phase consists in taking a snapshot of the other robots positions using its visibility sensors. In the compute phase, a robot computes a target destination among its neighbors, based on the previous observation. The move phase simply consists in moving toward the computed destination using motion actuators. We consider an asynchronous computing model, \emph{i.e.}, the ratio between the speed of the fastest robot and that of the slowest robot is finite but unbounded (and unknown to the robots), however each cycle is considered atomic (robots may nevertheless execute cycles concurrently). This model is referred to as the semi-synchronous ATOM model in the literature~\cite{SY99}. When cycle phases may not be atomic, it is possible for a robot to observe another robot while it moves, or to perform the  computing  (and  moving)  phase  with  an outdated observation. This non-atomic but still asynchronous model is known as the CORDA model~\cite{Pre01}. Of course, all executions in the ATOM model are also valid in the CORDA model. So, impossibility results for the ATOM model remains true in the CORDA model, and protocols for the  CORDA model are also valid for the ATOM model, but the converse is not true. 

In addition to the temporal uncertainty resulting from system asynchrony, the lack of a common coordinate system leads to a second kind of uncertainty which is spatial: there is no common notion of distance (robots do not share a common metric system) or direction (robots do not have a common compass). Another aspect of spacial uncertainty stems from the common assumption that an adversary has the ability to stop a robot movement before it reaches its planned destination. This implies that the distance effectively travelled by a robot at each cycle of operation is unpredictable.  

Finally, the robots that we consider here have weak capacities: they are \emph{anonymous} (they execute the same protocol and have no mean to distinguish themselves from the others), \emph{oblivious} (they have no memory that is persistent between two cycles). In some  problem instances such as gathering, robots may share the same position, which is called a multiplicity point. The ability for a robot to detect multiplicity is crucial to  solve some particular tasks. We  distinguish \emph{weak} and \emph{strong}  multiplicity  detection.   The  weak  multiplicity  detector detects whether there  is zero, one or more than one  robot at a particular location.  The strong  multiplicity  detector senses  the  exact number  of robots at a particular location.

\paragraph*{Fault-tolerance and Wait-freedom.}

As the output of an individual robot is its movement, faults in robot networks are characterized by the possibilities allowed for unexpected behavior. The most simple fault is the halting (or crash) fault (a robot simply stops moving forever). A halting fault-tolerant (or simply fault-tolerant) robot protocol permits robots that do not crash (that is, the correct robots) to properly complete a given task (such as gathering). The \emph{wait-freedom} property is the strongest non-blocking guarantee in ``classical'' distributed computing~\cite{Herlihy91}, as a wait-free algorithm guarantees that every execution completes in a finite number of steps, even if halting faults or simply adversarial scheduling occur. Simply put, in the context of robot networks, an arbitrary and unexpected delay observed at one robot may not prevent other robots from making progress toward the solution (in our case, gathering), even if $n-1$ robots are delayed or crashed ($n$ being the number of robots).  

Another kind of fault is the transient fault (that is, a fault of arbitrary nature that places the robot in some arbitrary state). Since we assume robots are oblivious (and do not remember their past states), a transient fault may simply put the robots in some arbitrary initial positions. A self-stabilizing robot protocol permits all robots to properly complete a given task after all transient faults are finished (that the whole set of robots has been placed in arbitrary locations). The most malicious kind of fault is that of Byzantine fault, which can make a robot move arbitrarily (both considering location and speed). 
  
\paragraph*{Related Work.}

The (fault-free) gathering problem was introduced in the seminal paper of~\cite{SY99} in the ATOM model. Deterministic gathering of $2$ oblivious robots was proved to be impossible to solve in a deterministic setting~\cite{SY99}, while deterministic gathering with at least $3$ robots was shown to be feasible both in the ATOM~\cite{SY99} and CORDA~\cite{cieliebak03} models. Randomization~\cite{defago3274fta} or adding persistent memory to robots permit to solve $2$-gathering both in ATOM~\cite{SY99} or CORDA model~\cite{BDPT10}.  

With the possibility of even a single Byzantine fault, gathering becomes impossible~\cite{AP06, IBTW11}, even when considering more than $2$ robots, the simple ATOM model, randomization capabilities, and persistent memory. So, positive results consider either weaker problems (such as convergence~\cite{BPT09,BDPT10}, that only requires robot to \emph{approach} a single point, rather than reaching it) with Byzantine faults, or weaker fault models for gathering. To our knowledge, only two works~\cite{AP06,DP09} consider deterministic gathering with faults in robot networks, both in the ATOM model. With respect to halting faults, Agmon and Peleg~\cite{AP06} solve gathering with at most one halting fault in a network where robots may not share the same chirality (that is, they may not agree about their handedness), yet assume that no two robots are located on the same position initially (so, the protocol is not self-stabilizing). Dieudonn\'{e} and Petit~\cite{DP09} present a self-stabilizing algorithm for gathering, again without chirality assumption, but assuming that there is no halting fault and that the number of robots is odd.

The gathering problem has also been studied in models with other limitations, such as robots having limited visibility~\cite{ando1999dmp,FlocchiniPSW05} or robots that are not dimensionless (i.e. they block both the motion and visibility of other robots)~\cite{CzyzowiczGP09}. Another scenario that has been studied is when robots are given the additional capability of using a directional compass which is however subject to inaccuracies and failure~\cite{SouissiDY06}.

\paragraph*{Our Contribution.}

We investigate the possibility of handling more than one halting fault in robot networks in a deterministic setting. In more details, we present a deterministic protocol for gathering that can handle up to $n-1$ halting faults (that is, the protocol is wait-free). It is known that deterministic gathering is impossible if the robots are equally distributed in two points on the plane (the so-called \emph{Bivalent} configuration). We consider protocols for robots that start from any arbitrary configuration other than the (impossible) Bivalent one. 
Thus, our protocol recovers from any transient faults as particular initial configurations are unnecessary. We use the ATOM model as in~\cite{AP06,DP09}, yet tolerate more halting faults and more potential initial configurations. The main additional assumption that we make is a common chirality for all robots that participate to the protocol. We also assume, as in~\cite{DP09}, a strong multiplicity detection mechanism. Our results are summarized in table~\ref{tab:results}.

\begin{table}
\centering
\begin{tabular}{|c|l|l|l|l|l|}\hline
\textbf{Reference} & \textbf{Model} & \textbf{Chirality} & \textbf{Multiplicity} & \textbf{Halting}& \textbf{Self-Stabilizing} \\
 &   &  &  \textbf{Detection} & \textbf{Faults}  &  \\\hline
\cite{AP06} & ATOM & No & Weak (binary) & $f \leq 1$ & No \\\hline
\cite{DP09} & ATOM & No & Strong & $f=0$ & If $n$ is odd  \\\hline
\textbf{This paper} & ATOM & Yes & Strong & $f < n$ & If not Bivalent  \\\hline
\end{tabular}
\caption{Resilience bounds for deterministic gathering}
\label{tab:results}
\end{table}

Our protocol is based on spatial invariants that are both simple to compute and are preserved by robot movements induced by the protocol. One of the most natural candidates for the case of gathering is the Weber point~\cite{W37}. Given a set of points $P$, the Weber point $c$ minimizes $\sum_{r \in P} distance(x, r)$ over all points $x$ in the plane. The Weber point has the key property of remaining unchanged under straight movements of any of the points towards or away from it. If the Weber point can be computed, it is simple to devise a robot protocol that solves gathering: all robots simply move toward the Weber point. Unfortunately, computing the Weber point is known to be difficult and was solved in special cases such as regular polygons~\cite{ACP05}, line~\cite{CT90}, and a number of symmetric and regular configurations~\cite{BL11}. A key result of this paper is a technique to compute the Weber point of a newly defined class of configurations, referred in the sequel as quasi-regular configurations, which are less symmetric than both symmetric and regular configurations.

Using this building block, the protocol for gathering can then be informally described as follows: \emph{(i)} if the configuration has some amount of symmetry, that is, it is symmetric, regular, or quasi-regular, then the robots move to the Weber point, \emph{(ii)} if the configuration is completely asymmetric, 
it is possible to unanimously elect one unique robot location and all robots may move towards this leader robot (or robots, in case of multiplicities). However these movements may not keep the leader invariant. Note that it is not possible to maintain a leader by moving individual robots in a preferred order (such a protocol would not be wait-free/crash-tolerant). Moreover, since the robots may be stopped by the adversary before reaching the destination, we risk forming the bivalent configuration from which the robots cannot recover and it would be impossible to achieve gathering. Thus, the algorithm involves several technicalities to ensure that the robots progress towards a gathered configuration in a wait-free manner, while avoiding the catastrophic \emph{bivalent} configuration.
  
\section{The Model and Notations}

\paragraph*{Robot Model:} 
 
There are $n$ robots modeled as points on a geometric plane. A robot can observe its environment and determine the location of other robots in the plane, relative to its own location. All robots are identical (and thus indistinguishable) and they follow the same algorithm. However each robot has its own local coordinates system and measure of unit distance (which may be distinct from that of other robots). The robots only share a common sense of handedness (i.e. they agree on the clockwise direction). Time is divided to discrete intervals called rounds and in each round a robot be either active or inactive. In each round, each active robot \textbf{r} makes exactly one step which consists of LOOK, COMPUTE and MOVE actions. During the LOOK stage, robot \textbf{r} gets a complete snap-shot containing the locations of every other robot in terms of the local coordinate system and unit distance used by robot \textbf{r}. (Note that multiple robots may occupy the same location in the plane and a robot can determine exactly how many robots are located at the same point.) During the COMPUTE stage, the robot executes its algorithm, using the snapshot as input and determines its next destination point. (Note that robots do not need to remember their previous steps.) During the MOVE stage, the robot moves towards the computed destination\footnote{If the computed destination is the current location, then the robot does not move.}. A move may end  before the robot reaches its destination. However there exists an (arbitrarily small) constant $\delta > 0$ such that if the destination point is closer than $\delta$,  the robot will reach it; otherwise, it  will move a distance of at least $\Delta$ towards its destination.
A robot that is inactive in a round does not take any actions during that round. Each (\emph{correct}) robot is active in infinitely many rounds. 

We denote the above model of computation as $ATOM[\diamond M]$ model (i.e. the ATOM model enhanced with strong multiplicity detection).

\paragraph*{Fault Model:}
We consider the crash fault model. A robot is \emph{faulty}
if there is a time at which it stops taking actions (\emph{i.e.} it crashes). 
However, a crashed robot remains visible to other robots in the system.
A robot that does not crash is \emph{correct}.
An algorithm is $f$-resilient if 
it works correctly when 
the number of faulty robots does not exceed $f < n$.
We may consider more general models of faults in which 
the adversary is allowed to fail subsets of robots that 
are not necessarily uniform.
An adversary can thus be characterized using its \emph{faulty sets} \cite{JM03}: the 
set of subsets of robots it is allowed to fail during an execution.
Sometimes, it is more convenient to describe adversaries using the set of their \emph{cores}.
A core \cite{JM03} is a minimal subset of robots that can not \emph{all} fail 
in any execution. For example, for the $f$-resilient adversary, 
any set of $f+1$ robots is a core.

\paragraph*{Notations:}

$\mathcal{R}=\{r_1, \ldots, r_n\}$ is the set of robots and $\mathcal{T}$ is the set of positive natural numbers, denoting time instances. At any time $\tau \in \mathcal{T}$, the configuration of the set of robots is given by the multiset $C_\mathcal{R}(\tau) = \{ p_1,\dots,p_n \}$ where each $p_i \in {\mathbb{R}^2}$. We shall drop the subscript $\mathcal{R}$ when it is obvious from context.
Let $\mathbb{P}$ be the set of all possible configurations of $n$ robots.
Formally, $\mathbb{P}={\mathbb{R}^2}^n$ where $\mathbb{R}$ is the set of real numbers.
%
A configuration is \emph{linear} if all its robots lie on the same line.

%
Given any robot $r$, $mul(r)$ denotes the multiplicity of the location occupied by $r$. 
Given a \emph{multiset} of robot positions $Q$, we denote by $U(Q)$ the corresponding \emph{set} of positions in $Q$ removing multiplicities (i.e. each point is U(Q) contains at least one robot). Given a configuration $C$, let $\secc(C)$ denote the smallest enclosing circle of the point set $U(C)$.
The center of a circle $G$ is denoted by $center(G)$. $CH(Q)$ denotes the convex hull of the points in $Q$.
%
%
%
%
Given two distinct points $u$ and $v$ on the plane, let $line(u,v)$ denote the straight line passing through these points and $(u, v)$ (resp. $[u, v]$) denote the open (resp. closed) interval containing all points in this line that lie between $u$ and $v$.
The half-line starting at point $u$ (but excluding the point $u$) and passing through $v$ is denoted by $\hfl{u}{v}$. Formally, $\hfl{u}{v} = \{ p \in line(u,v), p\neq u : v \in [u,p] \vee p \in [u,v] \}$. With respect to some point $c\in {\mathbb{R}^2} \setminus \{u, v\}$, the angle in the clockwise direction between line segments $[c,u]$ and $[c,v]$ is denoted by $\sphericalangle(u, c, v)$. The Euclidean distance between $u$ and $v$ is denoted by $|u,v|$.


\section{Symmetries in Robot Configurations}

\label{sec:symmetricity}

\subsection{Some Definitions}

Configurations may exhibit several kinds of symmetry.
In this section, we consider a specific form of symmetry called rotational symmetry
which we define in a precise sense and show how to quantify it. This is based on the concept of \emph{views}~ \cite{SY99}, as described below.


\begin{definition}[Views]
Let $C=\{p_1, \ldots, p_n\}$ 
be a configuration of robots.
%
Given a position $p \in U(C)$, 
define
the view of $p$,
denoted
$\mathcal{V}(p)$, 
as the expression of $C$
in the polar coordinate system 
whose center is $p$ and whose point $(1, 0)$ is defined as follows.
Let $c=center(\textsc{sec}(C))$.
If $(c \neq p)$, then $(1, 0)=c$. 
Otherwise $(1, 0)$ 
is any point $x\neq p \in U(C)$ that maximizes $\mathcal{V}(x)$.


\end{definition}

Note that in the definition above, the point $(1,0)$ is not uniquely defined, however the view of any point $p \in C$ is uniquely defined. Based on the definition of views, we can define an equivalence relation 
$\backsim_r$ on the set of robot locations, as follows: $\forall u, u^\prime \in U(C), (u \backsim_r u^\prime) \Leftrightarrow (\mathcal{V}(u)=\mathcal{V}(u^\prime))$.
The corresponding equivalence class for $u$ is denoted by $[u]_r$.
%
%
The following definition formalizes the notion of rotational symmetricity.
\begin{definition}[Rotational Symmetricity]
\label{def:rsymmetricity}
The (rotational) \emph{symmetricity}  of a configuration $C$, denoted $sym(C)$, 
is the cardinality of the biggest 
equivalence class defined by $\backsim_r$ on $U(C)$. 
That is, $sym(C)=max\{|[u]_r|~|~ u \in U(C)\}$.
\end{definition}

Our definition of symmetricity differs slightly from the one presented in \cite{DFSY10, SY99} but only for those configurations that contain either points of multiplicity or a point located in the center of their SEC.
%
The following result follows from the definition of symmetricity.

\begin{lemma}
\label{lem:rsymmetryPolygon}
Let $C$ be a configuration with $k=sym(C)>1$ and let $c=center(\secc(U(C)))$.
For every $u \in U(C)$ with $(u \neq c)$, it holds that $[u]_r$ a $k$-gon with center $c$ and
whose corners have the same multiplicity.
\end{lemma}

%
%
%


\paragraph*{Regularity:}

We now define some weaker forms of symmetry called \emph{regularity} and \emph{quasi-regularity} and show how to compute a unique gathering point (the so-called Weber point) in such configurations. If we consider any circle $G$ that encloses the points in a configuration $C$, we may order the points by sweeping the circle $G$ in a clockwise direction and  ordering points on the same radius w.r.t. their distance from the center. This idea leads to the following definitions (extending the concepts in \cite{ACP03,K05}).



\begin{definition} 
\label{def:successor}
Let $C=\{p_1, \ldots, p_n\}$ be a configuration and let $c \in \mathbb{R}^2$.
\begin{itemize}
\item {[{Successor}]} The \emph{clockwise} successor of $p_i \in C$ around $c$,
denoted by $S(p_i, c)$ is equal to the point $p_j \in C$ defined as follows:

\begin{itemize}
\item 
Let $X=\{p_k \in C ~|~  (p_k = p_i) \wedge (k < i) \}$.
If $X \neq \emptyset$, then $p_j=\argmax{p_k \in X}~k$.
\item 
Otherwise, let $Y=\{p_k \in C \cap (c, p_i) \}$.
If $Y \neq \emptyset$, then $p_j=\argmin{p_k \in Y} ~|p_i, p_k|$.

\item 
Otherwise, 
let $Z=\{ p_k\in C ~|~ 
(\nexists p \in C: 
0 < \sphericalangle(p_i,c,p) < \sphericalangle(p_i,c,p_k)\}$.
In this case, 
$p_j=\argmax{p_k \in Z} ~(|c, p_k|, k)$.
\end{itemize}

\item 
{[{$k$-th Successor}]} The $k$-th successor of $p_i$ around $c$, denoted $S^k(p_i, c)$ 
is defined recursively as follows:
If $k>1$, $S^k(p_i, c)=S(S^{k-1}(p_i, c))$; and $S^1(p_i, c)=S(p_i, c)$.

\end{itemize}

\end{definition}

The \emph{string of angles} around $c$ started in $p_i$, denoted by $SA(p_i, c)$
is the string $\alpha_1 \ldots \alpha_m$ such that $m=n-mul(c)$ and
$\alpha_i=\sphericalangle(S^{i-1}(p_i), c, S^i(p_i))$. 
The \emph{size} of $SA(p_i, c)$, denoted by $|SA(p_i, c)|$ is equal to $m$.
%
A string $SA$ is $k$-periodic if it can be written as $SA=x^k$ where $1 \leq k\leq |SA|$.
The greatest $k$ for which $SA$ is $k$-periodic is called \emph{the periodicity} of $SA$ and is denoted
by $per(SA)$.

\begin{figure}[htb!]
\begin{center}
\includegraphics[trim=10cm 5cm 16cm 2.5cm, clip=true, width=0.32\textwidth]{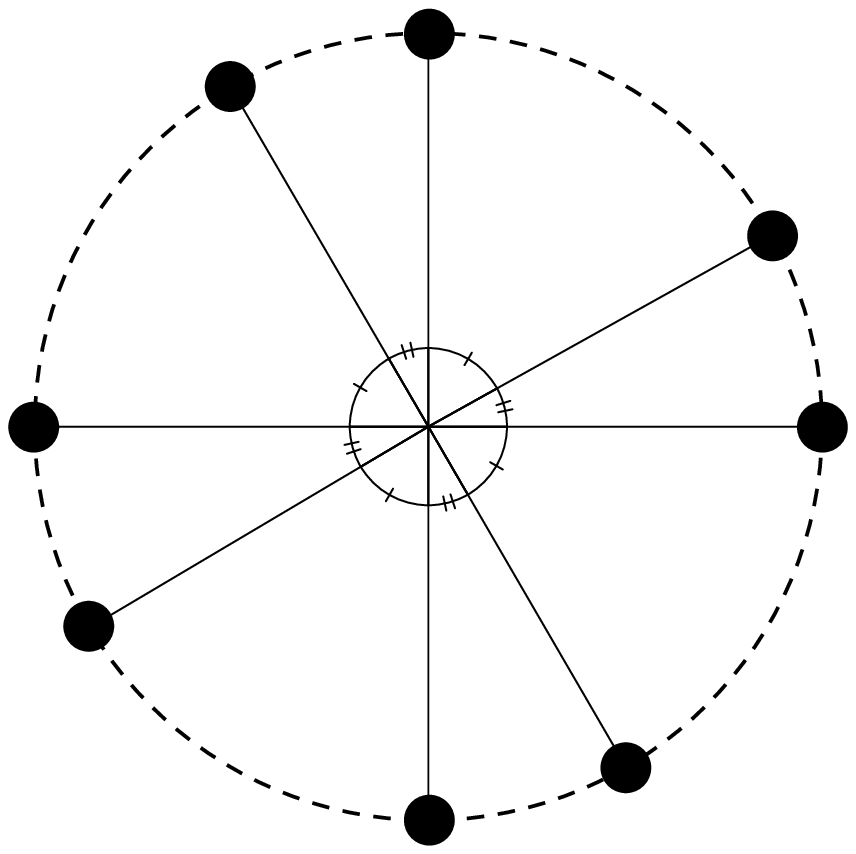}
\includegraphics[trim=10cm 5cm 16cm 2.5cm, clip=true, width=0.32\textwidth]{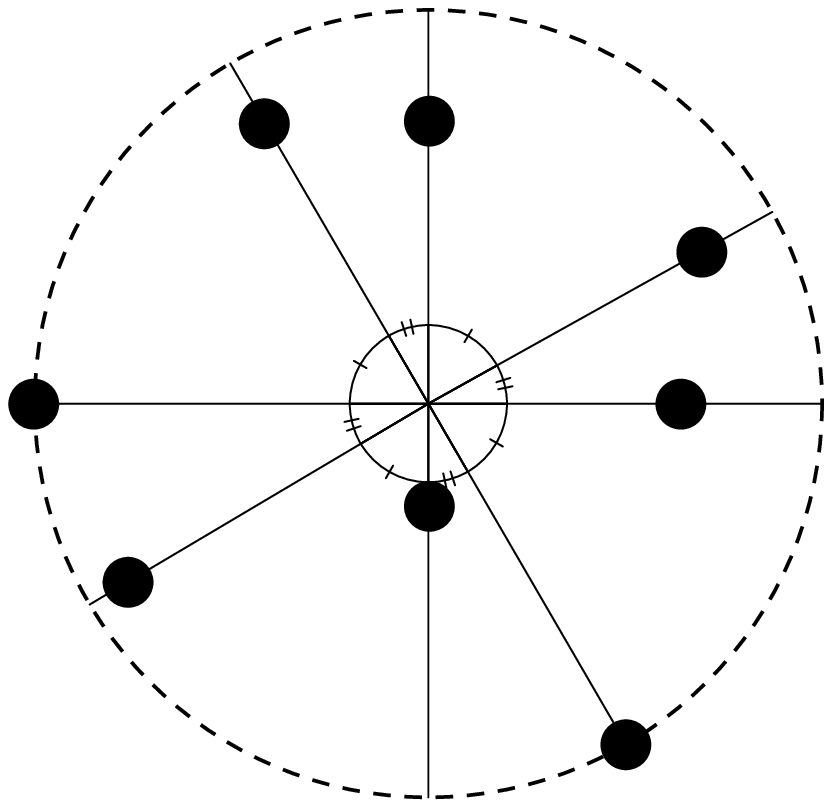}
\includegraphics[trim=10cm 5cm 16cm 2.5cm, clip=true, width=0.32\textwidth]{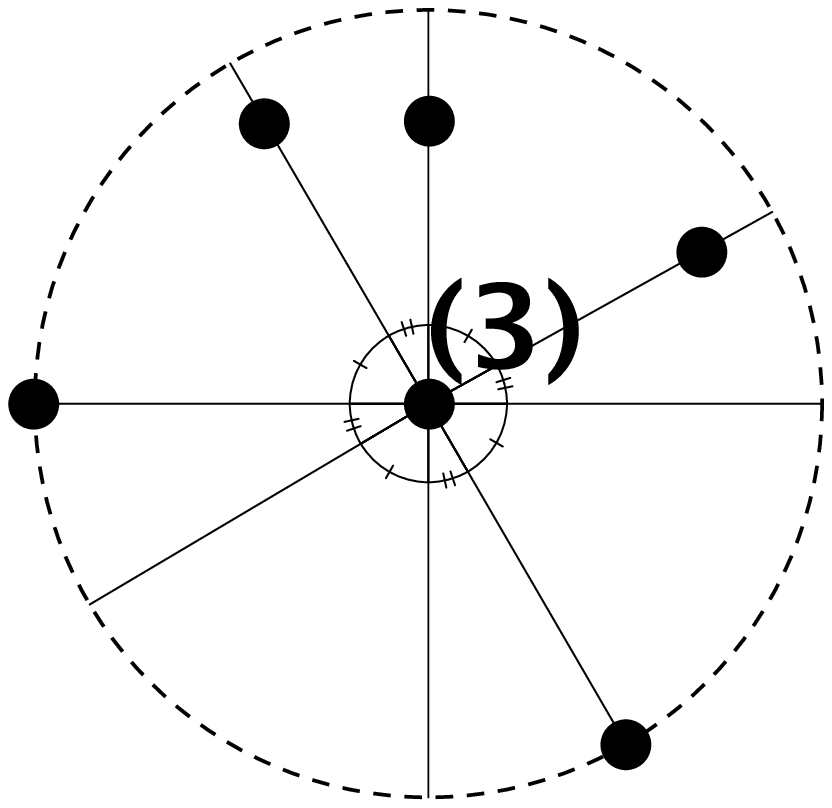}
\end{center}
\caption{Configurations that are (i) Symmetric with $sym(C)=4$, (ii) Regular with $reg(C)=4$, (iii) Quasi-Regular with $qreg(C)=4$. The numbers in parentheses represent the multiplicity of a point.} \label{fig:conf}
\end{figure}

\begin{definition}[Regularity]
\label{def:regularity}
A configuration $C$ of $n$ points is regular if there exists a point $c \in \mathbb{R}^2$ and $\exists m>1$ such that $per(SA(c))=m>1$.
In this case, the \emph{regularity} of $C$, denoted $reg(C)$, is equal to $m$.
Otherwise, $reg(C)=1$.
The point $c$ is called \emph{the center} of regularity and is denoted by $CR(C)$.
\end{definition}

\begin{definition}[Quasi Regularity]
A configuration $C$ of $n$ points is \emph{quasi-regular} (or Q-regular)
iff there exist (1) a point $c  \in \mathbb{R}^2$ 
and (2) a regular configuration $C^\prime$ with center of regularity $c$ 
which can be obtained 
from $C$ be moving only points located at $c$ if any. 
Formally, $C$ is quasi-regular with center $c \in \mathbb{R}^2$ 
iff $\exists C^\prime \in \mathbb{P}$ such that $reg(C^\prime)>1$, $CR(C^\prime)=c$ and
$\forall p \in C^\prime \setminus C, p=c$.
%
In this case, the quasi-regularity of $C$, denoted $qreg(C)$ = $reg(C^\prime)$ and the \emph{center of quasi-regularity} denoted $CQR(C)=c$. If $C$ is not quasi-regular then $qreg(C)=1$.
\end{definition}

Note that each configuration that is symmetric is also regular. More precisely, $(sym(C)>1) \Rightarrow (reg(C)=sym(C))$. Each regular configuration is also quasi-regular (with $C^\prime=C$).

\begin{definition}[Weber Points]
\label{def:WeberPoint}
The Weber points of a configuration $C$, denoted $WP(C)$, are the set of points that 
minimize the sum of distances with points of $C$.
Formally, $WP(C)=\argmin{x \in \mathbb{R}^2} \sum_{i=1}^{n} |x, p_i|$.
\end{definition} 

Non-linear configurations are known to have a unique Weber point while linear configurations may have infinitely many Weber points. The Weber points of a linear configuration $C$ are points in the interval $[min(Med(C)), max(Med(C))]$, where $Med(C)$ the set of median points. If a linear configuration $C$ has a single median, then this point is the unique Weber point $WP(C)$. We now show how to compute the Weber-point of some special non-linear configurations. 


%
%
%

\subsection{Computation of Weber Points}

Given a set of points $C$ and any arbitrary point $p \in \mathcal{R}^2$, we define $SEC(C,p)$ as the smallest circle centered at $c$ that encloses all points in $C$. Given a point $x \in C$ and some $\alpha \in [0, 2\pi]$,
the successor of $x$ with respect to point $p$ and angle $\alpha$, denoted by $S(x, p, \alpha)$ is the point $y$ such that $|p, x|=|p, y|$  and $\sphericalangle(x, p, y)= \alpha$. 

\begin{lemma}\label{lem:WPinv}
Let $C=\{p_1, \ldots, p_n\}$ and $C^\prime=\{p_1^\prime, \ldots, p_n^\prime\}$
two configurations.
Let $X = \{x \in WP(C) ~|~ \forall i \in [1, n]: p_i^\prime \in [p_i, x]\}$.
If $X \neq \emptyset$ then $WP(C^\prime)=X$.
\end{lemma}

\begin{proof}
Let $Y=\{x \in \mathbb{R}^2 ~|~ \forall i \in [1, n]: p_i^\prime \in [p_i, x]\}$.
Note that $X= Y \cap WP(C)$.

Observe that:
$$\sum_{i=1}^{n} |x, p_i^\prime| = \sum_{i=1}^{n} |x, p_i| + \sum_{i=1}^{n} (|x, p_i^\prime|- |x, p_i|)$$

By definition, the points of $WP(C)$ are those points $x$ that minimize $\sum_{i=1}^{n} |x, p_i|$.
Moreover, the points of $Y$ are those that minimize $\sum_{i=1}^{n} (|x, p_i^\prime|- |x, p_i|)$.
Hence, the points of $X=Y\cap WP(C)$ minimize the two sums and minimize their sum also.
It follows that the points that minimize $\sum_{i=1}^{n} |x, p_i^\prime|$ are those in $X$.
Thus, $WP(C^\prime)=X$.
\end{proof}

\begin{corollary}
\label{cor:linearUniqueWP}
 If $C$ is a configuration with a unique Weber point $c$
 and if $C^\prime$ is a configuration that is obtained from $C$
 by moving robots towards $c$, then 
 the Weber point of $C^\prime$ is also unique and is equal to $c$.
\end{corollary} 
 
In the following Lemma we show that
the center of quasi-regularity of a configuration is also its Weber point.
 
\begin{lemma}
\label{lem:CQR=WP}
For every non-linear configuration $C$ that is quasi-regular, $CQR(C)=WP(C)$.
\end{lemma}

\begin{proof}
Let  $c=CQR(C)$ and $qreg(C)=k>1$. We have to prove that $c=WP(C)$.
By definition of quasi-regularity, there exists a regular configuration $C^\prime$ whose center of regularity is $c$ and which can be obtained from $C$ by moving only points located at $c$ if any. 
Seen in the reverse sense, $C$ can obtained from $C^\prime$ 
under straight movement of points towards $c$.
Hence, by Corollary \ref{cor:linearUniqueWP}, $WP(C)=c$ if and only if $WP(C^\prime)=c$.

Therefore, it suffices to prove that the center of regularity 
of any configuration $C^\prime$ is also
its Weber point. 
This claim was already proved in \cite{BL11} but for a slightly different definition of regularity 
that does not allow the presence of multiplicity points nor points lying on $c$.
However, the argument is the same and does not depend on these factors. 
We reproduce it here for completeness.
Let $C^{\prime\prime}$ be a configuration obtained from $C^\prime$ as follows:
For each point $p \in C$ not located at $c$, move $p$ towards the point
that is at the intersection of $\hfl{c}{p}$ and $SEC(C,c)$.
Clearly, the obtained configuration  $C^{\prime\prime}$ is symmetric
with center of symmetricity $c$ and $sym(C^{\prime\prime})=reg(C^\prime)$.
Note that $C^\prime$ can be obtained from $C^{\prime\prime}$ only by moving points in $C^{\prime\prime}$ towards $c$.
Again, by Corollary \ref{cor:linearUniqueWP} implies that
$WP(C^{\prime\prime})=c$ implies $WP(C^\prime)=c$.

To finishes our proof, it suffices then to show that the center of R-symmetricity of any configuration $C^{\prime\prime}$ is also its Weber point.
Assume for contradiction that $WP(C^{\prime\prime}) = c^\prime \neq c$. 
Let $k=sym(C^{\prime\prime}) > 1$.
Let $P_k$ by the regular $k$-gon whose center is $c$ and one of whose vertices is $c^\prime$.
By symmetricity, if $c^\prime$ is a Weber point of $C^{\prime\prime}$, so are all points of $P_k$.
But the Weber point of a non-linear configuration is necessarily unique - A contradiction!

\end{proof}

The following lemma proves a property about regular configurations: 
 
\begin{lemma}
\label{lem:propertyRegular}
Let $C$ be any configuration and let $p \in \mathbb{R}^2$.
Let $m>1$.
$C$ is regular with center $p$ and $reg(C)=m$ iff
$\forall x \neq p \in \mathbb{R}^2: \forall k \in [1,m]: \forall y= S(x, p, \frac{2k\pi}{m}):$
it holds that $\hfl{p}{x}$ and $\hfl{p}{y}$ contain the same number of robots of $C$.
\end{lemma}

\begin{proof}
$C$ is regular with $reg(C)=m$ and center $p$ iff there exists a configuration $C^\prime$ that is symmetric
with the same center of regularity and 
such that $C$ can be obtained from $C^\prime$ by moving robots 
of $C^\prime$ towards $p$ without reaching it.
Moreover $sym(C^\prime)=m$.
 
Fix $x, y \in \mathbb{R}^2$ such that 
$y=S(x, p, \frac{2k\pi}{m})$ for some $k \in [1, m]$.
We have to prove that $\hfl{p}{x}$ and $\hfl{p}{y}$ contain the same number of robots of $C$.
But observe that the number of robots in $\hfl{p}{x}$ and $\hfl{p}{y}$
remains invariant when we transform $C^\prime$ into $C$ as 
robots are allowed to move only towards $p$ without reaching it.
So no robot joins or leaves either $\hfl{p}{x}$ or $\hfl{p}{y}$.
Hence,to prove our claim it suffices to show that $\hfl{p}{x}$ and $\hfl{p}{y}$ contain the same number of robots of $C^\prime$.

Since $sym(C^\prime)=m$, it follows that $C^\prime$ remains invariant if we rotate it 
around $p$ with an angle of $\frac{2k\pi}{m}$.
Note that $\hfl{p}{x}$ can be seen as the result of rotating $\hfl{p}{y}$ by an angle of $\frac{2k\pi}{m}$.
Hence $\hfl{p}{x}$ contains the same number of robots of $C^\prime$ as $\hfl{p}{y}$.
This proves the lemma.
\end{proof}

\begin{definition}
Let $C$ a configuration with $|U(C)|>1$ and let $p \in C$. Let $m>1$.
We define $X_m(C, p)$ as the following set of points 
$\{ x \in \textsc{sec}(C, p)  ~|~ \exists k \in [1, m]: \exists y \in S(x, p, \frac{2k\pi}{m}): (p, y] \text{ contains at least one robot}\}$.
For each point $x \in X_m(C, p)$, 
let $\textsc{loc}(C, x, p)$ (or $\textsc{loc}(C, x))$ 
denote the number of robots of $C$ that are located in $(p, x]$
and let $\textsc{obj}(C, x)$ denote
$max \{\textsc{loc}(C, y) |  (y = S(x, p, \frac{2k\pi}{m})) \wedge (k \in [1, m]) \}$.
\end{definition}

\begin{lemma}
\label{lem:testQRegular}
Given a configuration $C$ and a point $p \in C$, $C$ is q-regular with center $p$
and $qreg(C)=m>1$ iff 
\begin{equation}
\label{eq:lemma}
\tag{$\alpha$}
mult(p) \geq \sum\limits_{x \in X_m(C, p)} (\textsc{obj}(C, x)-\textsc{loc}(C, x))
\end{equation}

\end{lemma}

\begin{proof}
According to Definition \ref{def:regularity}, $C$ is q-regular with center $p$ and $qreg(C)=m$ iff
(i) there exists a configuration $C^\prime$ that is regular with center $p$, 
(ii) $reg(C^\prime)=m$
and 
(iii) $C$ can be transformed into $C^\prime$ by moving only robots located at $p$.
To prove the lemma it suffices to show that 
$(i) \wedge (ii) \wedge (iii) \Leftrightarrow (\alpha)$.

\begin{description}

\item[$\Leftarrow)$]

Assume $(\alpha)$ holds. 
For each $x \in X_m(C, p)$, we move $(\textsc{obj}(C, x)-\textsc{loc}(C, x))$ robots from $p$ to $x$.
Since $(\alpha)$ is satisfied, there are enough robots located in $p$ to perform this action.
Let $C^\prime$ be the resulting configuration.
This proves (iii).
Note that $X_m(C^\prime, p)=X_m(C, p)$.

By construction of $C^\prime$, it holds that
$\forall x \in X_m(C^\prime, p): \textsc{loc}(C^\prime, x)= \textsc{obj}(C, x)$.
But by definition of  $\textsc{obj}(C, x)$, we have
$\forall x \in X_m(C, p): \forall k \in [1, m]: \forall y = S(x, p, \frac{2k\pi}{m}): 
\textsc{obj}(C, x)=\textsc{obj}(C, y)$.
It follows that 

$$\forall x\in X_m(C^\prime, p): \forall k \in [1,m]: \forall y= S(x, p, \frac{2k\pi}{m}): \textsc{loc}(C^\prime, x)=\textsc{loc}(C^\prime, y)$$

Consequently, according to Lemma \ref{lem:propertyRegular}, it holds that (i) 
$C^\prime$ is regular with center $p$ and (ii) $reg(C^\prime)=m$.

\item[$\Rightarrow)$]

Assume $(i) \wedge (ii) \wedge (iii)$.
Since $C^\prime$ can be obtained from $C$ by moving only robots located at $p$ according to (iii),
it follows that $X_m(C, p) \subseteq X_m(C\prime, p)$.
Moreover, 
\begin{equation}
\forall x \in X_m(C, p): \textsc{loc}(C, x) \leq \textsc{loc}(C^\prime, x)
\end{equation}

Since $C^\prime$ is regular with center $p$ and $reg(C^\prime)=m$, it 
holds according 
to Lemma \ref{lem:propertyRegular} that
$\forall x \in X_m(C, p): \forall k \in [1,m]: \forall y \in S(x, p, \frac{2k\pi}{m}):
\textsc{loc}(C^\prime, x)=\textsc{loc}(C^\prime, y)$.
But $\textsc{loc}(C^\prime, y) \geq \textsc{loc}(C, y)$ 
according to Equation ().
Hence $\textsc{loc}(C^\prime, x) \geq \textsc{loc}(C, y)$.
It follows that:

$$
\forall x \in X_m(C, p):
\textsc{loc}(C^\prime, x) \geq \textsc{obj}(C, y)
$$

Hence,
$$
\forall x \in X_m(C^\prime, p):
(\textsc{loc}(C^\prime, x) - \textsc{loc}(C, x)) \geq (\textsc{obj}(C, y) - \textsc{loc}(C, x))
$$

But all the robots that are in $(\textsc{loc}(C^\prime, x) - \textsc{loc}(C, x))$ moved there from $p$.
Consequently:

$$
 mult(p) \geq \sum\limits_{x \in X_m(C, p)} (\textsc{obj}(C, x)-\textsc{loc}(C, x))
$$

\end{description}

\end{proof}

We now state the main result of this section. 

\begin{theorem}
\label{thm:detectionWeberPoint}
Given a non-linear configuration $C$ of $n$ robots,
there exists an algorithm that detects if $C$ is quasi-regular and if so it outputs its 
center of q-regularity $CQR(C)$.
\end{theorem}

\begin{proof}
As shown in Lemma \ref{lem:CQR=WP}, $CQR(C)=WP(C)$ hence it is unique.
If $WP(C) \in C$, then it can be found by applying Lemma
\ref{lem:testQRegular} as follows:
for each $p \in C$ we test $p$ is the center of q-regularity of $C$.
Otherwise, $WP(C) \not\in C$ which means that $C$ is regular.
Consequently, $WP(C)$ can be computed as shown in \cite{BL11}.
\end{proof}

\section{Configurations}

\subsection{Classes of Configurations}

In the gathering algorithm, robots compute their next destinations based on the current configuration. Before presenting the algorithm, we present a classification of the robot configurations which will simplify the algorithm description. In the following,  we formally define six classes of configurations and prove
that they constitute a partition of the set $\mathbb{P}$ of all possible configurations of $n$ robots.

\begin{description}

\item[Bivalent($\mathcal{B}$)] 
$\mathcal{B}= \{ C \in \mathbb{P}~|~ (\forall u \in U(C): mul(u)=n/2)\}$.
$\mathcal{B}$ is the set of configurations where the robots are 
equally distributed over two points in the space.

\item[Multiple ($\mathcal{M}$)] 
$\mathcal{M} = \{  C \in \mathbb{P}~|~ \exists u \in U(C): \forall v\neq u \in U(C): mul(v) < mul(u)\}$.
A configuration $C$ belongs to $\mathcal{M}$
if it has a point $u$ whose multiplicity is greater than that of any other distinct point in $C$.


\item[Colinear($\mathcal{L}$)] 
$\mathcal{L} = \{  C \in \mathbb{P}~|~ ( C \text{ is linear} )\wedge (C \not\in \mathcal{B} \cup \mathcal{M})\}$.
We define the subsets $\mathcal{L}1\mathcal{W}$ and $\mathcal{L}2\mathcal{W}$ of colinear configurations depending on whether 
their Weber point is unique or not.
That is, $\mathcal{L}1\mathcal{W} = \{  C \in \mathcal{L} ~|~ (WP(C)\text{ is unique} )\}$ and 
$\mathcal{L}2\mathcal{W} =  \mathcal{L}  \setminus \mathcal{L}1\mathcal{W}$.

\item[Q*Regular ($\mathcal{QR}$)]
$\mathcal{R} = \{  C \in \mathbb{P}~|~ (qreg(C)>1) \wedge
 (C \not\in \mathcal{B} \cup \mathcal{M} \cup \mathcal{L}) \}$.

%

\item[Asymmetric ($\mathcal{A}$)]
$\mathcal{A} = \{  C \in \mathbb{P}~|~ (sym(C)=1)\wedge
(C \not\in \mathcal{B} \cup \mathcal{M} \cup \mathcal{L}  \cup  \mathcal{QR}\}$.

\end{description}

Let $\mathbb{X}=\{\mathcal{B} ,  \mathcal{M} , \mathcal{L} , \mathcal{QR}, \mathcal{A}\}$.
It is easy to see that $\mathbb{X}$ is a partition of $\mathbb{P}$. By definition, the classes are mutually disjoint. All linear configurations belong to the set $\mathcal{B} \cup \mathcal{M} \cup \mathcal{L}$. For a non-linear configuration $C$ either $sym(C)>1$ which implies $C \in \mathcal{QR} \cup \mathcal{B} \cup \mathcal{M}$, or  $sym(C)=1$ which implies that $C \in \mathcal{A} \cup \mathcal{B} \cup \mathcal{M}$. Thus $\bigcup \mathbb{X}= \mathbb{P}$.



\subsection{Properties of Configurations}

\begin{lemma}\label{lem:linear}
Let $C$ be a linear configuration. 
The following properties hold:
\begin{enumerate}
\item $(|U(C)|=2) \Rightarrow (C \in \mathcal{B} \cup \mathcal{M})$
\item $(|U(C)|=3) \Rightarrow (C \in \mathcal{M} \cup \mathcal{L}1\mathcal{W})$
\item $(C \in \mathcal{L}2\mathcal{W}) \Rightarrow (|U(C)| \geq 4)$
\end{enumerate}
\end{lemma}

\begin{proof}

\begin{enumerate}
\item Assume $|U(C)|=2$.
That is, $U(C)$ consists in two distinct points $u_1$ and $u_2$.
If $mul(u_1)=mul(u_2)$, then $C \in \mathcal{B}$.
Otherwise, $C \in \mathcal{M}$
as either $(mul(u_1) > mul(u_2))$ or 
$(mul(u_2) > mul(u_1))$.

\item 
Assume $|U(C)|=3$, \emph{i.e.} $U(C)$ consists in three distinct points, let them be $u_1, u_2, u_3$.
Suppose w.l.o.g. that $u_2 \in [u_1, u_3]$.
We assume that $C \not\in \mathcal{L}1\mathcal{W}$ and we prove that
$C \in \mathcal{M}$.
The fact that $|U(C)|=3$ implies that $C \not\in \mathcal{B}$.
Since $C$ is linear and $C \not\in \mathcal{L}1\mathcal{W} \cup \mathcal{B}$, 
it follows from the definition
of $\mathcal{L}$ that $C \in \mathcal{L}2\mathcal{W} \cup \mathcal{M}$.
To prove that $C \in \mathcal{M}$ it suffices then to show that 
$C \not\in \mathcal{L}2\mathcal{W}$.
Assume towards contradiction that $C \in \mathcal{L}2\mathcal{W}$.
This means that the set $Median(C)$ is not a singleton.
Hence, there are at least two points in $U(C)$ that are in $Median(C)$.
Consequently, either $u_1$ or $u_3$ belongs to $Median(C)$ (together with $u_2$).
Assume w.l.o.g that $u_1 \in Median(C)$.
This implies that $mul(u_1) \geq \lceil n/2 \rceil$.
Hence $mul(u_2)+mul(u_3) \leq n - \lceil n/2 \rceil$.
That is, $mul(u_2)+mul(u_3) \leq \lceil n/2 \rceil$.
Since $mul(u_2) \geq 1$ and $mul(u_3) \geq 1$, 
it follows that
$(mul(u_2) \leq \lceil n/2 \rceil - 1)$ and
$(mul(u_3) \leq \lceil n/2 \rceil - 1)$.
But we showed that $mul(u_1) \geq \lceil n/2 \rceil$.
Consequently we have
$mul(u_2) < mul(u_1)$ and
$mul(u_3) < mul(u_1)$.
This means that $C \in \mathcal{M}$ which contradicts our assumption that
$C \in \mathcal{L}2\mathcal{W}$.
This finishes the proof of $C \in \mathcal{M}$ (assuming $|U(C)|=3$ and $C \not\in \mathcal{L}1\mathcal{W})$.
Hence:

$$(|U(C)|=3) \Rightarrow (C \in \mathcal{M} \cup \mathcal{L}1\mathcal{W})$$

\item This follows from the above two results.

\end{enumerate}
\end{proof}


\begin{definition}[Safe points]
\label{def:freePoints}
Given a configuration $C$, a robot position $p \in C$ is \emph{safe}
iff $\forall q\in \mathcal{R}^2 \setminus \{p\}$:
$\hfl{p}{q}$ contains at most $(\lceil n/2 \rceil -1)$ robots of $C$.
\end{definition}

The notion of \emph{safe points} is important because any safe point can be used as a gathering point without the possibility that the robots form the bivalent $\mathcal{B}$ configuration while moving towards it. We can show the following properties for safe points.

\begin{lemma}\label{lem:safe}
Any non linear configuration contains a safe point.
\end{lemma}

\begin{proof}
Let $C$ be a non linear configuration.
We say that $Q \subseteq C$ is a quorum iff:
(i) $|Q| \geq \lfloor n/2 \rfloor +1$ and
(ii) all points of $Q$ are collinear and $Q$ is maximal for this property, that is,
for any $Q^\prime \supset Q$, the points of $Q^\prime$ are not collinear.
Let $\textsc{line}(Q)$ denote the line in which are located the points of $Q$.

Let $Q_1$ and $Q_2$ any two \emph{distinct} quorums of $C$.
Condition (i) implies that $Q_1$ and $Q_2$ intersect, \emph{i.e.} $Q_1 \cap Q_2 \neq \emptyset$
and the maximality condition in (ii) implies that
$\textsc{line}(Q_1) \neq \textsc{line}(Q_2)$.

We show in the following that any point that is not safe belongs necessarily to a quorum.
Let $p \in C$ that is not free. We prove the existence of quorum $Q$ to which $p$ belongs.
Since $p$ is not safe, there exists $q \in  \mathcal{R}^2 \setminus \{p\}$ such that
$\hfl{p}{q}$ contains at least $(\lceil n/2 \rceil)$ robots located in it. 
Hence, $p \cup \hfl{p}{q}$ contains at least $\lceil n/2 \rceil+1$ robots.
Let $S$ denote the multiset of positions of these robots.
Since $(|S| \geq \lceil n/2 \rceil+1)$ and the points of $S$ are collinear,
there exists a set $Q$ with $Q \supseteq S$ that is a quorum with $p \in Q$.
 


We prove the lemma by contradiction. 
Assume that no point of $C$ is safe, i.e. each point belongs to a quorum.
This implies, as $C$ is not linear, that there are at least two distinct quorums
because a single quorum cannot contain all elements of $C$, otherwise the configuration would be linear.
Let ${Q}_1$ and $Q_2$ be any two quorums of $C$ with $Q_1 \neq Q_2$
and let $p$ be a point in $Q_1 \cap {Q}_2$. 
Since $p$ is not safe according to the contradiction assumption, 
there exists some 
$q \in  \mathcal{R}^2 \setminus \{p\}$ such that
$\hfl{p}{q}$ contains at least $\lceil n/2 \rceil$ robots positions. 
Denote by $X$ the multiset containing these positions.
Note that $|X|\geq \lceil n/2 \rceil$

As ${Q}_1 \neq Q_2$, it follows 
according to the maximality of property (ii) of quorums
 that $\textsc{line}(Q_1) \neq \textsc{line}(Q_2)$.
Hence, either $line(p,q) \neq \textsc{line}(Q_1)$ or $line(p,q) \neq \textsc{line}(Q_2)$.
Assume \emph{w.l.o.g} that $line(p,q) \neq \textsc{line}(Q_1)$.
Since $p \in line(p,q)$, $p \in Q_1$ and $line(p,q) \neq \textsc{line}(Q_1)$
it follows that $line(p,q) \cap \textsc{line}(Q_1)=\{p\}$.
Hence, the robots positions that are in $\hfl{p}{q}$ do not belong to $Q_1$
which means that $X \cap Q_1=\emptyset$.
Hence, $|X \cup Q_1|=|X|+|Q_1| \geq (\lceil n/2 \rceil) + (\lfloor n/2 \rfloor +1)$.
That is $|X \cup Q_1| \geq n+1$.
But $|X \cup Q_1| \subseteq C$, a contradiction!
Thus the lemma holds.
\end{proof}

\begin{lemma}
\label{lem:nosafe_LB}
If $C \in \mathcal{B} \cup \mathcal{L}2\mathcal{W}$, then $C$ does not have a safe point.
\end{lemma}

\begin{proof}
Assume towards contradiction that there exists a position $p \in C$ that is safe.
Since $C$ is linear, this means that $|\{ q \in C~|~q<p\}| \leq
(\lceil n/2 \rceil -1)$ 
and $|\{ q \in C~|~q>p\}| \leq
(\lceil n/2 \rceil -1)$.

But $C \in \mathcal{B} \cup \mathcal{L}2\mathcal{W}$, 
it follows that $C$ has two distinct median positions, let them be $u_1$ and $u_2$
and assume that $u_1 < u_2$.
It holds that either $(u_2 > p)$ or $(u_1 < p)$.
Assume \emph{w.l.o.g.} that $(u_1 <p)$.
Since $u_1$ is a median position in $C$, it holds that
$|\{ q \in C~|~q \leq u_1\}| \geq
\lceil n/2 \rceil$.
Hence, as  $u_1 <p$, it follows that
$|\{ q \in C~|~q<p\}|  \geq
\lceil n/2 \rceil$; 
A Contradiction!
\end{proof}


\section{The Algorithm}


The following lemma is a simple generalization of Lemma 3.1 in \cite{AP06}
that takes into account configurations containing multiplicity points and
general adversaries characterized using their cores.
%
Given a configuration $P$, an algorithm $\mathcal{A}$,
denote by $M(P, \mathcal{A})$ the set of positions of robots that
$\mathcal{A}$ instructs to move in $P$~\cite{AP06}.

\begin{lemma}
A convergence or gathering algorithm $\mathcal{A}$ is tolerant against an adversary $\mathcal{X}$
\emph{only if} at each configuration $P$, either (1) $M(P, \mathcal{A})$ is a superset of a core of $\mathcal{X}$
or (2) $|U(P \setminus M(P, \mathcal{A}))| \leq 1)$.
\end{lemma}

\begin{proof}
The proof is similar to that of \cite{AP06}.
Consider a configuration of $P$.
Let $M$ and $\overline{M}$ denote respectively 
the subsets
$M(P, \mathcal{A})$ and $P \setminus M(P, \mathcal{A})$.
Assume for contradiction that 
(1) 
$M$ is not  a superset of a core of $\mathcal{X}$, \emph{i.e.} $M$ is a faulty set
and (2) $|U(\overline{M})| \geq 2$.
Hence, $\mathcal{X}$ is allowed to fail the robots of $M$ when 
the current configuration is $P$ (we assume that all robots of $\overline{M}$ are correct).
Since no robot in $\overline{M}$ is allowed to move,
the next configuration is identical to $P$.
Therefore, the system will remain in the configuration $P$ indefinitely.
But $|U(\overline{M})| \geq 2$, 
thus the robots of $\overline{M}$ remain indefinitely separated from each others 
and no convergence nor gathering can be ever achieved.
This contradicts the assumption of $\mathcal{A}$ being a convergence or gathering algorithm.

\end{proof}

As a consequence, since we want our algorithm to be wait-free ($(n-1)$-tolerant), it must be the case that at each configuration C, there is at \emph{most} one location $c\in U(C)$ such that the robots at $c$ are allowed to stay in the same position when activated, while all other robots must choose a destination different from the one they are currently occupying. 
The algorithm must also ensure that robots never reach the configuration $\mathcal{B}$, due to the following impossibility result.

\begin{lemma}\cite{DP09}
Starting from a configuration of type $\mathcal{B}$, there is no algorithm that achieves gathering even 
in fault-free $ATOM[\diamond M]$ model.
\end{lemma}

We now define more precisely the objective of a fault-tolerant gathering algorithm.
At any time $\tau$ during the execution of the algorithm, we define $F(r, \tau)=true$ if robot $r$ has crashed at time $t_i \leq \tau$. The set of non-faulty robots at time $\tau$ is denoted by $Live(\mathcal{R},\tau)$ $=$ $\{ r_i \in \mathcal{R} | F(r_i, \tau)=false \}$.

\begin{definition}
\label{def:predicateGathered}
Given a set of robots $\mathcal{R}$ that form configuration $C$ at time $\tau$, $\textsc{gathered}(\mathcal{R}, \tau)$ = true iff 
$(|U(Live(\mathcal{R},\tau))|=1)$ and 
$(M(C, \mathcal{A}) \cap U(Live(\mathcal{R},\tau)) = \emptyset$.
\end{definition}

%


\begin{figure}[htb]
\centering{ \fbox{
\begin{minipage}[t]{150mm}
\scriptsize
\renewcommand{\baselinestretch}{2.5} \resetline{}
\begin{tabbing}
aaaaa\=aa\=aaa\=aaa\=aaaaa\=aaaaa\=aaaaaaaaaaaaaa\=aaaaa\=\kill 

\textbf{Input:} $C$ (The observed configuration during the precedent \textsc{look} phase). \\
\textbf{Output:} The destination of the robot. 
~\\

\textsc{compute}(): \\
~\\
\line{P1}  \> $r \leftarrow$ \textsc{My position in} $C$ \\
~\\

\line{M1}   \> \textbf{if} $C \in \mathcal{M}$ \textbf{then} \\
\line{M2}   \>\>  $elected \leftarrow \argmax{p \in C} ~mul(p)$ \\ 
\line{M3}   \>\> \textbf{if} $(r=elected) \vee (\not\exists p \in C: p \in (r, elected))$ \textbf{then} \\
\line{M4}   \>\>\> \textbf{return} $elected$ \\
\line{M5}   \>\> \textbf{else} \\
\line{M6}   \>\>\> $X \leftarrow \{p \in C: p \not\in \hfl{elected}{r} \}$ \\
\line{M7}   \>\>\> $v \leftarrow \argmin{p \in X} ~ (k ~|~ p=S^k(r, elected))$ \\
\line{M8}   \>\>\> \textbf{let} 
$d \in \mathbb{R}^2$ \emph{s.t.} $((|d, elected| = |r, elected|) \wedge (\sphericalangle(r, elected, d)=\sphericalangle(r, elected, v)/3)) $ \\
\line{M9}   \>\>\>  \textbf{return} $d$\\
~\\

\line{QR1}   \> \textbf{if} $C \in \mathcal{QR} \cup \mathcal{L}1\mathcal{W}$ \textbf{then} \\ 
\line{QR2}   \>\> \textbf{return} $WP(C)$ \\
~\\


\line{A1}   \> \textbf{if} $C \in \mathcal{A}$ \textbf{then}\\
\line{A2} \>\> $X \leftarrow $the set of \textbf{safe} points in $U(C)$. \\
\line{A3} \>\> $elected \leftarrow \argmax{p \in X}~ (mul(p), \frac{1}{\sum_{q \in C} dist(p, q)}, \mathcal{V}(p))$ \\

\line{A4}  \>\> \textbf{return} $elected$ \\
~\\

\line{L1} \> \textbf{if} $C \in \mathcal{L}2\mathcal{W}$ \textbf{then} \\
\line{L2}  \>\> $c \leftarrow center(C)$ \\ 
\line{L3}  \>\> \textbf{if} $r \not \in CH(C)$ \textbf{then}  \\
\line{L4}  \>\>\> \textbf{return} $c$ \\
\line{L5}  \>\> \textbf{else} \\
\line{L6}  \>\>\>
\textbf{let} 
$d \in \mathbb{R}^2$ \emph{s.t.} 
$(|d, c| = |r, c|) \wedge (\sphericalangle(r, c, d)=\Pi/4)$ \\
\line{L7}  \>\>\> \textbf{return} $d$ \\


\end{tabbing}
\normalsize
\end{minipage}
}
\caption{Gathering Algorithm: \textsc{compute} Phase}
\label{alg:main}
}
\end{figure}

\subsection{Gathering Algorithm}

We now describe the algorithm in terms of actions taken by a robot $r$ based on the current configuration and the position of the robot $r$ within the configuration. A more technical description is given in Figure \ref{alg:main}.

\paragraph*{\textbf{Configuration $C\in \mathcal{M}$}\\}

Let $c$ be the unique point of maximum multiplicity in $C$. If robot $r$ is located at $c$, it does not move. Otherwise, if there are no robots between $r$ and $c$, robot $r$ moves directly towards $c$ and if not, it does a side-step i.e. it moves to the closest point $d$ on a half-line $\hfl{c}{d}$ such that the angle between half-line $\hfl{c}{d}$ and half-line $\hfl{c}{r}$ is less than or equal to 1/3 of the angle between half-line $\hfl{c}{p}$ and half-line $\hfl{c}{r}$, for any other robot location $p \in C$. This ensures the robot does not collide with another robot, i.e. it does not create a new point of maximum multiplicity. Note that there may be multiple robots colocated with $r$, these robot may make the same move as $r$. However the value of $mul(r)$ would never increase unless $r$ reaches the point $c$. Thus, the algorithm ensures that the robots remain in a configuration of type $\mathcal{M}$ until gathering is achieved.

\paragraph*{\textbf{Configuration $C\in \confLW{1} $}\\}

By definition, we know that configuration $C$ contains a unique Weber-point $c$ which is also the median and can be computed easily. Each robot $r$ moves directly towards the Weber point $c$ which remains invariant during the movement. Eventually the configuration changes to $\mathcal{M}$ or a gathered configuration.

\paragraph*{\textbf{Configuration $C\in \mathcal{QR}$}\\}

In this case, robot $r$ moves to the center $c$ of quasi-regularity of $C$ (which is also the Weber-point). Thus, the Weber-point $c$ remains invariant during the movement and eventually the configuration changes to $\mathcal{M}$ or a gathered configuration.

\paragraph*{\textbf{Configuration $C\in \mathcal{A}$}\\}

Since $C$ is not linear, we know that there exists a safe point in $U(C)$. When there are multiple safe points, the algorithm selects a unique point $c$ from among the safe points in U(C). This is always possible since the configuration is asymmetric (i.e. the view of each point is unique). The algorithm chooses the point $c$ based on the multiplicity(c), the sum of distances of all other robots to c, and finally the view of c (in this order, and maximizing the first parameter, minimizing the second parameter and maximizing the third parameter). Each robot $r$ moves towards this unique point $c$. We will show that the configuration $C'$ obtained after one step of the algorithm is of type $\mathcal{M}$, $\mathcal{QR}$, $\confLW{1}$, or $\mathcal{A}$ (but not $\mathcal{B}$ or $\confLW{2}$). Further if the next configuration is again of type $\mathcal{A}$ then either the maximum multiplicity increases or the minimum sum of distance decreases. This ensures that the algorithm converges towards a configuration of type $\mathcal{M}$ or a gathered configuration.     

\paragraph*{\textbf{Configuration $C\in$ $\confLW{2} $}\\}

In this case, there are at least $4$ distinct points in the configuration. The algorithm instructs the robots at the two end-points of the line to move away from the line. Any robot that is not located in one of the end-points is instructed to move towards the center of the line. If any of the robots at the end-points move then the next configuration would be non-linear and thus, the algorithm switches to one of the other cases above. Otherwise, if the robots are the end-points never move (i.e. they are crashed) then the configuration remains linear but the sum of distances between correct robots decreases, and the robots would eventually converge to a gathered configuration or a configuration of type $\mathcal{M}$. 

%
%
%
%
%


\subsection{Proof of Correctness}

We now show that starting from any configuration except the bivalent configuration $\mathcal{B}$, the algorithm described in Figure~\ref{alg:main} eventually forms a gathered configuration. The proof is divided into several parts, each dealing with configurations of a different type. 

\subsubsection{Configurations of type $\mathcal{M}$}

\begin{lemma}
\label{lem:confMconverge}
Let $C_(\tau) \in \mathcal{M}$.
There exists a time $\tau^\prime > \tau$ such
that $\textsc{gathered}(\mathcal{R}, \tau)$=true.
\end{lemma}

\begin{proof}
Let $elected(\tau)$ be the destination chosen by the robots in configuration $C(\tau)$, i.e. $elected(\tau)=\argmax{p \in C(\tau)} ~mul(p)$. A robot position $p \in C(\tau)$ 
is said to be \textbf{free} with respect to $elected(\tau)$ if no robot is located in the interval $(p, c)$.
The lemma follows from the following two claims that we prove below:

\begin{description}
\item[C1:] 
$(C(\tau) \in \mathcal{M}) \Rightarrow 
((C(\tau+1) \in \mathcal{M})
\wedge 
(elected(\tau+1)=elected(\tau)))$

\item[C2:] 
($\forall \tau_i \geq \tau: (C(\tau_i) \in \mathcal{M})
\wedge 
(elected(\tau_i)=elected(\tau))$)
$\Rightarrow$
($\exists \tau^\prime \geq \tau: \textsc{gathered}(\mathcal{R}, \tau^\prime)$=true).

\end{description}

\paragraph{Proof of C1:}
Let $c=elected(\tau)$ be the point of maximum multiplicity in the configuration $C(\tau)$. We need to show that $c$ remains the point of maximum multiplicity in $C(\tau+1)$. In fact we show a stronger result that no two robots that were in distinct locations at $\tau$ can be at the same location at time $\tau+1$, unless the robots are at $c$. Let us assume the contrary, i.e. let $r_1$ and $r_2$ be robots that occupied distinct locations in $C(\tau)$ but occupy the same location $p\neq c$ in $C(\tau+1)$. Note that neither of the robots $r_1$ and $r_2$ are located at $c$ at time $\tau$ (since otherwise the algorithm would instruct them to remain at $c$ and they would not be at $p$ at $\tau+1$). 

According to the algorithm any robot $r$ in configuration $C(\tau) \in \mathcal{M}$ can make two possible moves: (i) either robot $r$ moves directly towards $c$ (Line (5) of algorithm) or, (ii) robot $r$ moves to a point $d$ such that $rcd$ is an isosceles triangle with central angle $0 < \theta < \pi/3$ at $c$ (Line (9) of algorithm). 
If both the robots $r_1$ and $r_2$ both make move of type (i), then they are distinct free points and in this case their paths may not intersect except at $c$. Otherwise, suppose one of the robots (say $r_1$) makes a move of type (ii) directly towards a point $d$. Consider the triangle $r_2cd$ and let $\sphericalangle(r_2, c, d)=\theta$. The other robot $r_2$ is located either on the half-line $\hfl{c}{r_1}$ or, on a different half-line \hfl{c}{r_2} which forms an angle greater than $3*\theta$ with the half-line \hfl{c}{r_1} w.r.t. point $c$. In the second case, the path of robot $r_2$ will never intersect the line segment between $c$ and $d$. In the first case, either robot $r_2$ is on a free point (and thus, it will move on the line segment $[c,r_2]$ which does not intersect the line segment $[c,d]$) or robot $r_2$ is not free and thus it makes a move on a line segment parallel to $[c,d]$. In both cases, there is no common point $p$ in the path of the two robots.

\paragraph{Proof of C2:}
In this case, if $c=elected(\tau)$ is the point of maximum multiplicity in $C(\tau)$ then $c$ is the unique point of maximum multiplicity in all subsequent configurations.
Whenever a robot on a free point is activated it moves closer to the point $c$. Whenever a blocked (i.e. not free) robot is activated, at least one robot moves from a blocked position to a free point. Once a robot $r$ moves to a free point at time $\tau_i$, it may be blocked in subsequent steps by only robots that moved with robot $r$ in that same time step (i.e. these robots were \emph{live} at that time step). Thus an adversary can prevent a non-faulty robot $r$ from reaching point $c$ only by changing a live robot to a crashed robot after each step in which robot $r$ is activated. After a finite time, the adversary will run out of live robots. Thus all live robots will eventually reach $c$. Once a robot reaches $c$, the algorithm never instructs the robot to move (since $c$ is the unique point of maximum multiplicity). Thus, $\textsc{gathered}(\mathcal{R}, \tau)$ will be true at that time.
\end{proof}

\subsubsection{Configurations of type $\mathcal{L}1\mathcal{W}$}

\begin{lemma}
\label{lem:confL1Wconverge}
Let $C(\tau) \in \mathcal{L}1\mathcal{W}$.
There exists a time $\tau_c > \tau$ such that
either $(C(\tau_c) \in \mathcal{M})$ or, 
$(\textsc{gathered}(\mathcal{R}, \tau_c)=true)$.
\end{lemma}

\begin{proof}
The lemma follows from the following two claims that we prove below:

\begin{description}
\item[C1:] 
$(C(\tau) \in \mathcal{L}1\mathcal{W}) \Rightarrow 
((C(\tau+1) \in \mathcal{M} \cup \mathcal{L}1\mathcal{W})
\wedge 
(WP(C(\tau+1))=WP(C(\tau))))$

\item[C2:] 
$(\forall \tau^\prime \geq \tau: (C(\tau^\prime) \in \mathcal{L}1\mathcal{W})
\wedge 
(WP(C(\tau^\prime))=WP(C(\tau))))
\Rightarrow
(\exists \tau_c \geq \tau: \textsc{gathered}(\mathcal{R}, \tau_c)=true)$.

\end{description}

\paragraph{Proof of C1:}
Since $C(\tau) \in \mathcal{L}1\mathcal{W}$, it follows that
$WP(C(\tau))=c$ is unique and 
$C(\tau+1)$ is obtained by moving robots in $C(\tau)$ towards $c$ (line \ref{QR2} of the algorithm).
Hence, according to Corollary \ref{cor:linearUniqueWP}, 
$WP(C(\tau+1))=WP(C(\tau))=c$.
Moreover, $C(\tau+1)$ is linear. This, combined with 
the fact that its Weber point is unique implies that $C(\tau+1)$ cannot
be of type $\mathcal{B}$ or $\mathcal{L}2\mathcal{W}$.
Therefore, $C(\tau+1) \in \mathcal{M} \cup \mathcal{L}1\mathcal{W}$.

\paragraph{Proof of C2:}
Whenever a robot in configuration $\confLW{1}$  is activated it moves towards the Weber-point $c$ and Weber-point remains invariant due to this movement. Thus, for all configurations $C(\tau^\prime)$ the Weber-point is the same point $c$. For each non-faulty robot $r$ the distance between $r$ and $c$ decreases every time the robot $r$ is activated (unless $r$ is already at $c$). Thus all non-faulty robots are at the point $c$ at some time $\tau_c$ and $\textsc{gathered}(\mathcal{R}, \tau_c)$=true.
\end{proof}

\subsubsection{Configurations of type $\mathcal{QR}$}

\begin{lemma}
\label{lem:confQRconverge}
Let $C(\tau) \in \mathcal{QR}$.
There exists a time $\tau_c > \tau$ such that
 $(C(\tau_c) \in \mathcal{M} \cup \mathcal{L}1\mathcal{W}) 
 \vee
(\textsc{gathered}(\mathcal{R}, \tau_c)=true)$.
\end{lemma}

\begin{proof}
Let $c=WP(C(\tau))$.
The lemma follows from the following two claims that we prove below:

\begin{description}
\item[C1:] 
$(C(\tau) \in \mathcal{QR}) \Rightarrow 
(C(\tau+1) \in   \mathcal{M} \cup \mathcal{L}1\mathcal{W} \cup\mathcal{QR})
\wedge 
(WP(C(\tau+1))=WP(C(\tau))))$

\item[C2:] 
$(\forall \tau^\prime \geq \tau: (C(\tau^\prime) \in \mathcal{QR}) 
\wedge 
(WP(C(\tau^\prime))=WP(C(\tau)))
\Rightarrow
(\exists \tau_c \geq \tau: \textsc{gathered}(\mathcal{R}, \tau_c)=true)$.

\end{description}

\paragraph{Proof of C1:}
Since $C(\tau) \in \mathcal{QR}$,
robots that are activated at $\tau$ move towards $WP(C(\tau))=CQR(C(\tau))$
according to line (\ref{QR2}) of the code.
Hence, since $C(\tau)$ is Q-regular, the obtained configuration $C(\tau+1)$ 
is Q-regular also with the same center of Q-regularity as $C(\tau)$
Hence, $WP(C(\tau+1))=CQR(C(\tau+1))=CQR(C(\tau))=WP(C(\tau))$.

As $C(\tau+1)$ is Q-regular, it holds according to the definition of configurations
$\mathcal{QR}$ that if $C(\tau+1) \not\in \mathcal{B} \cup \mathcal{M} \cup \mathcal{L}$
then $C(\tau+1) \in \mathcal{QR}$.
Therefore, $C(\tau+1) \in \mathbb{B} \cup \mathcal{M} \cup \mathcal{L} \cup \mathcal{QR}$.
It remains to show that $C(\tau+1) \not \in \mathcal{B} \cup \mathcal{L}2\mathcal{W}$.
For this, it suffices to show that $C(\tau+1)$ has a unique Weber point.
But this follows Corollary \ref{cor:linearUniqueWP} and the fact that $WP(C(\tau))$ is unique.


\paragraph{Proof of C2:}
Since all the configurations after $\tau$ are of type $\mathcal{QR}$ and since the Weber point
remains invariant after $\tau$, it follows that all activated robots after $\tau$ choose the same
destination point: $WP(C(\tau))$.
Hence, there is a time $\tau_c$ at which all live robots have reached this point.
That is, $\textsc{gathered}(\mathcal{R}, \tau_c)=true$.
\end{proof}

\subsubsection{Configurations of type $\mathcal{A}$}

\begin{lemma}
\label{lem:confAconverge}
Let $C(\tau) \in \mathcal{A}$.
There exists a time $\tau_c > \tau$ such that
 $(C(\tau_c) \in \mathcal{M} \cup \mathcal{L}1\mathcal{W} \cup \mathcal{QR}) 
 \vee
(\textsc{gathered}(\mathcal{R}, \tau_c)=true)$.
\end{lemma}

\begin{proof}
Given a configuration $C$, 
let $\phi(C)$ be the couple of values defined by 
$(mult, sum)=max \{(mul(p), \frac{1}{\sum_{q \in C} |p, q|}) ~|~p \in C\}$.

The lemma follows from the claims C1 and C3 below. Claim C2 is used to prove C3.

\begin{description}
\item[C1:] 
$(C(\tau) \in \mathcal{A}) \Rightarrow 
(C(\tau+1) \in   \mathcal{M} \cup \mathcal{L}1\mathcal{W} \cup \mathcal{QR} \cup \mathcal{A})
$

\item[C2:] \[ \begin{array}{r c l}
     (C(\tau) \in \mathcal{A})  & \Rightarrow &  (C(\tau+1) = C(\tau)) \\
                                             & \vee &  (\phi(C(\tau+1)).mult > \phi(C(\tau)).mult)  \\
				         & \vee &  (\phi(C(\tau+1)).mult = \phi(C(\tau)).mult) \wedge (\phi(C(\tau+1)).sum^{-1} < \phi(C(\tau)).sum^{-1})) \\

   \end{array}
\]

\item[C3:] 
$(\forall \tau^\prime \geq \tau: C(\tau^\prime) \in \mathcal{A}) 
\Rightarrow
(\exists \tau_c \geq \tau: \textsc{gathered}(\mathcal{R}, \tau_c)=true)$.

\end{description}

\paragraph{Proof of C1:}
$C(\tau) \in \mathcal{A}$ means that $sym(C(\tau))=1$.
Hence, each position in $U(C(\tau))$ has a unique view.
This guarantees that the \enquote{elected} position computed in line \ref{A3} is unique
and common to all activated robots at $\tau$, let us denote it by $u$.
Moreover, $u$ is safe in $C(\tau)$.
Hence, all activated robots at $\tau$ move towards the same safe point $u$
which results in configuration $C(\tau+1)$.
We observe that $u$ is also safe in $C(\tau+1)$ since for all $x \in \mathbb{R}^2 \setminus \{u\}$, 
the number of robots that are located at $\hfl{u}{x}$ does not increase between $\tau$ and $\tau+1$ 
(it may even decrease if some of them reach $u$).
Thus, $C(\tau+1)$ contains at least one safe point ($u$).
According to Lemma \ref{lem:nosafe_LB} this implies that
$C(\tau+1) \not\in \mathcal{B} \cup \mathcal{L}2\mathcal{W}$ 
which suffices to prove the claim.

\paragraph{Proof of C2:}
Assume that $C(\tau+1) \neq C(\tau)$.
Let $u \in C(\tau)$ be the common \enquote{elected} position
chosen by the algorithm (line \ref{A3}).
Hence, by definition, $(mul(u), \sum_{p_i(\tau) \in C(\tau)} |u, p_i(\tau)| )=\phi(C(\tau))$.
Since all activated robots at $\tau$ move to the same 
destination $u$, it follows that the resulting configuration
$C(\tau+1)$ satisfies:
$$\sum_{p_i(\tau+1) \in C(\tau+1)} |u, p_i(\tau+1)|  \leq \sum_{p_i(\tau) \in C(\tau)} |u, p_i(\tau)| = \phi(\tau).sum^{-1}$$

Since $(C(\tau+1) \neq C(\tau))$, there exists at least one robot $r_i$ whose position at $\tau+1$ 
is distinct from its position at $\tau$.
That is $p_i(\tau+1) \neq p_i(\tau)$.
Note that since all robots move towards $u$, it follows that 
$p_i(\tau+1) \in [p_i(\tau), u]$.
We distinguish between two cases:
\begin{enumerate}
\item $p_i(\tau+1)=u$. In this case $mul(u)$ is incremented.
Note that $u$ is still safe in $\tau+1$ (as shown in the proof of C1).
Hence, $\phi(\tau+1).mult = mul(u)^{\tau+1}> \phi(\tau).mult$

\item $p_i(\tau+1) \neq u$. That is, $r_i$ is stopped by the scheduler before it reaches $u$.
But since 
the scheduler guarantees to each robot to move by a distance of at least $\Delta$ 
before it can stop it, it follows that $(|u, p_i(\tau+1)| \leq  |u, p_i(\tau)|-\Delta)$
which, combined with the above inequality gives:
$$\sum_{p_i(\tau+1) \in C(\tau+1)} |u, p_i(\tau+1)|  \leq (\sum_{p_i(\tau) \in C(\tau)} |u, p_i(\tau)|)- \Delta=\phi(C(\tau)).sum^{-1}-\Delta$$
Note that the multiplicity of $u$ does not decrease even if no robot reaches it, \emph{i.e.}
$mul(u)^{\tau+1} \geq \phi(C(\tau)).mult$.
Note that $u$ is still safe in $\tau+1$ (as shown in the proof of C1).
Hence, 
\[ \begin{array}{r c l}  \phi(C(\tau+1)) 
 &  \geq (mul(u), \frac{1}{\sum_{p_i(\tau+1) \in C(\tau+1)} |u, p_i(\tau+1)|}) \\
 &  \geq (\phi(\tau).mul, \frac{1}{\phi(\tau).sum^{-1}-\Delta}) \\
 \end{array} \]

Therefore, either $(\phi(\tau+1).mul > \phi(\tau))$
or  $((\phi(\tau+1).mul=\phi(\tau)) \wedge (\phi(\tau+1).sum^{-1} \leq \phi(\tau).sum^{-1}-\Delta)$.
\end{enumerate}

This proves the claim.

\paragraph{Proof of C3:}
Assume that $(\forall \tau^\prime \geq \tau: C(\tau^\prime) \in \mathcal{A})$.
We have to prove that:
$$\exists \tau_g \geq \tau: \forall \tau^\prime \geq \tau_g: (C(\tau^\prime))=C(\tau))$$
That is, the configuration does not change after $\tau_g$ which implies the existence of 
a time after $\tau_g$ at which all the live robots lie on the same position.
That is
$\exists \tau_c \geq \tau_g: \textsc{gathered}(\mathcal{R},\tau_c)=true$.

The claim follows from claim C2 above. There exists a time $\tau_1 \geq \tau$ after which $\phi().mult$ cannot increase
since the multiplicity of points is upper bounded by $n$.
Moreover, there exists a time $\tau_2 \geq \tau_1$ after which $\phi().sum^{-1}$ cannot decrease since the sum of distance is lower bounded by $0$.
Hence, according to claim C2, after time $\tau_2$, the configuration remains the same
and the claim follows by setting $\tau_c=\tau_2$.

\end{proof}


\subsubsection{Configurations of type $\mathcal{L}2\mathcal{W}$}

\begin{definition}
Assume that $C(\tau)$ is linear.
Let $u^{-}(\tau)$ and $u^{+}(\tau)$ denote
$min(U(C(\tau)))$ and $max(U(C(\tau)))$ respectively.
Denote by $S^{-}(\tau)$ and $S^{+}(\tau)$
the set of robots located at $u^{-}(\tau)$ and $u^{+}(\tau)$ respectively
and let $S^{0}(\tau) = \mathbb{R} \setminus S^{-}(\tau) \cup S^{+}(\tau)$.
\end{definition}

If $C(\tau) \in \mathcal{L}2\mathcal{W}$,
then Lemma~\ref{lem:linear} implies that
$|U(C(\tau))| \geq 4$.
Hence, 
the sets $S^{-}(\tau)$, $S^{0}(\tau)$ and $S^{+}(\tau)$ 
in this case are non empty and 
pairwise disjoint.

\begin{lemma}
\label{lem:L2WnotB}
If $C(\tau) \in \mathcal{L}2\mathcal{W}$, then $C(\tau+1) \not\in \mathcal{B}$.
\end{lemma}

\begin{proof}
It suffices to show that $|U(C(\tau+1))| \geq 3$.
This simply follows from the fact that robots of $S^{-}(\tau)$, $S^{0}(\tau)$ and $S^{+}(\tau)$ 
occupy distinct positions and these groups remain disjoint when the robots activated at $\tau$ move
towards their destinations (computed in line \ref{L2} for $S^{0}(\tau)$ and line \ref{L6} 
for  $S^{-}(\tau)$ and $S^{+}(\tau)$).
\end{proof}

\begin{lemma}
\label{lem:Ct1_notL2M}
Assume $C(\tau) \in \mathcal{L}2\mathcal{W}$.
If at least one robot in $S_{-}(\tau) \cup S^{+}(\tau)$
is activated at $\tau$,
then $C(\tau+1) \not\in \mathcal{L}2\mathcal{W}$.
\end{lemma}

\begin{proof}
Let $a$ and $b$ denote $u^{-}(\tau)$ and $u^{+}(\tau)$ respectively and 
let $c$ be the midpoint of $[a,b]$. Due to Lemma~\ref{lem:linear} we know that
$|U(C(\tau))| \geq 4$ and thus other than $a$ and $b$, there exists at least two other points in $U(C)$. The scenario considered in this lemma can be partitioned into the following three cases: (i) No robot located at $a$ are activated at step $\tau$ (ii) No robot located at $b$ are activated at step $\tau$ (iii) At least one robot from each of $a$ and $b$ are activated. We will show that in each case, $C(\tau+1) \not\in \mathcal{L}2\mathcal{W}$. Let $L = line(a,b)$. Note that all robots lie on $L$ at time $\tau$.

\paragraph{Case (i):} 
In this case, at least one robot $r$ located at point $b$ is activated and according to the algorithm, the robot $r$ moves towards a point $p$ such that $\sphericalangle(b, c, p)=\pi/4$. The new position $p^\prime$ reached by the robot $r$ lies in $(b,p)$ and thus, $p^\prime \notin L$. Note that any robot $r^\prime \in S^{0}(\tau)$, still remains on line $L$ at some point distinct from $a$ (since robots in $S^{0}(\tau)$ are allowed to move only towards $c \in L$). Thus, $U(C(\tau+1))$ contains the points $p^\prime \notin L$, $a \in L$, and at least one other point in $L$ that is distinct from $a$. Hence $C(\tau+1)$ is not linear, which implies that $C(\tau+1) \notin \mathcal{L}2\mathcal{W}$.

\paragraph{Case (ii):}
This case is exactly symmetrical to case (i) above and the same result holds.

\paragraph{Case (iii):}
If not all the robots located at $a$ and $b$ are activated at time $\tau$ then we can use similar arguments as above to show that the configuration $C(\tau+1)$ is not linear. Thus the only interesting case to consider is when all robots at $a$ move to the same location $a^\prime$ and all robots at $b$ move to the same location $b^\prime$. Note that $a^\prime \notin L$ and $b^\prime \notin L$ and $a^\prime \neq b^\prime$. Thus, $line(a^\prime, b^\prime)$ is distinct from line $L$. However all the robots $\in S^{0}(\tau)$ must remain on line $L$ in step $(\tau+1)$. If the configuration $C(\tau+1)$ is linear then all robots $\in S^{0}(\tau)$ must be located on the same point at step $(\tau+1)$ and this point must be the point of intersection of $L$ and $line(a^\prime, b^\prime)$. In other words, $|U(C(\tau+1))|=3$, which implies that $C(\tau+1) \notin \mathcal{L}2\mathcal{W}$ (due to Lemma~\ref{lem:linear}).
\end{proof}


\begin{lemma}
\label{lem:confL2Wconverge}
Let $C(\tau) \in \mathcal{L}2\mathcal{W}$.
There exists a time $\tau_c > \tau$ such that
 $(C(\tau_c) \in \mathcal{M} \cup \mathcal{L}1\mathcal{W} \cup \mathcal{QR} \cup \mathcal{A}) 
 \vee
(\textsc{gathered}(\mathcal{R},\tau_c)=true)$.
\end{lemma}

\begin{proof}
Assume for the sake of contradiction that 
$\forall \tau^\prime \geq \tau: (C(\tau^\prime) \in \mathcal{L}2\mathcal{W} \cup \mathcal{B}) 
\wedge (\textsc{gathered}(\mathcal{R},\tau^\prime)=false)$.

But since $C(\tau) \in \mathcal{L}2\mathcal{W}$ 
and Lemma \ref{lem:L2WnotB} says that a configuration of type $\mathcal{B}$ cannot 
come after a configuration of type $\mathcal{L}2\mathcal{W}$, it follows that: 

$$\forall \tau^\prime \geq \tau: (C(\tau^\prime) \in \mathcal{L}2\mathcal{W}) 
\wedge (\textsc{gathered}(\mathcal{R},\tau^\prime)=false)$$

According to Lemma \ref{lem:Ct1_notL2M}, this is only possible if  the robots
located at the endpoints are never activated, which means that they are all faulty.
Hence, the center of the configuration $c=center(C(\tau))$ remains constant during all the execution and 
all correct robots eventually reach this point (line \ref{L2} of the algorithm).
Hence, there is a time  $\tau_c > \tau$ at which all correct robots are 
located at $c$. Thus, $\textsc{gathered}(\mathcal{R},\tau_c)=true$.
A contradiction.
\end{proof}



\bibliographystyle{plain}


\end{document}